\documentclass[a4paper]{article}

\usepackage[margin=1in]{geometry}
\usepackage[utf8]{inputenc}
\usepackage{amsmath,amssymb,amsthm}
\usepackage[textsize=scriptsize,textwidth=2.7cm]{todonotes}
\usepackage{booktabs}
\usepackage{enumitem}
\usepackage{xspace,mdframed}
\usepackage{nicefrac}
\usepackage{tikz}
\usepackage{pgfplots}
\pgfplotsset{compat=1.18}

\newcommand{\e}{\textsf{e}}
\newcommand{\eps}{\varepsilon}
\newcommand{\alg}{\textsc{Alg}\xspace}

\newcommand{\calE}{\ensuremath{\mathcal{E}}}
\newcommand{\calM}{\ensuremath{\mathcal{M}}}

\newtheorem{theorem}{Theorem}

\newtheorem{proposition}{Proposition}
\newtheorem{observation}{Observation}
\newtheorem{corollary}{Corollary}

\allowdisplaybreaks

\title{Threshold Testing and Semi-Online Prophet~Inequalities} 

\author{Martin Hoefer\thanks{Goethe University Frankfurt, Germany. Email: \texttt{mhoefer@em.uni-frankfurt.de}. Supported by DFG Research Unit ADYN (411362735) and grant 514505843.} \and Kevin Schewior\thanks{University of Southern Denmark, Denmark. Email: \texttt{kevs@sdu.dk}. Supported by the Independent Research Fund Denmark, Natural Sciences, grant DFF-0135-00018B.}}

\date{}

\begin{document}

\maketitle

´
\begin{abstract}
We study \emph{threshold testing}, an elementary probing model with the goal to choose a large value out of $n$ i.i.d.\ random variables. An algorithm can \emph{test} each variable $X_i$ once for some threshold $t_i$, and the test returns binary feedback whether $X_i \ge t_i$ or not. Thresholds can be chosen adaptively or non-adaptively by the algorithm. Given the results for the tests of each variable, we then select the variable with highest conditional expectation. We compare the expected value obtained by the testing algorithm with expected maximum of the variables.

Threshold testing is a semi-online variant of the gambler's problem and prophet inequalities. Indeed, the optimal performance of \emph{non-adaptive} algorithms for threshold testing is governed by the standard i.i.d.\ prophet inequality of approximately $0.745 + o(1)$ as $n \to \infty$. We show how \emph{adaptive} algorithms can significantly improve upon this ratio. Our adaptive testing strategy guarantees a competitive ratio of at least $0.869 - o(1)$. Moreover, we show that there are distributions that admit only a constant ratio $c < 1$, even when $n \to \infty$. Finally, when each box can be tested multiple times (with $n$ tests in total), we design an algorithm that achieves a ratio of $1-o(1)$.
\end{abstract}

\section{Introduction}

Consider an application process in which $n$ job candidates are interviewed sequentially one by one for a single position. For each candidate, we assume the qualification for the job can be expressed by an i.i.d.\ non-negative random variable $X_i$ with known distribution $F$. The goal is to maximize the expected value of the selected candidate. To which extent is the optimal achievable value harmed by the online arrival of the candidates? This is the classic \emph{gambler's problem}, in which the loss in expected value is expressed by prophet inequalities~\cite{KrengelS77,Lucier17,CorreaFHOV18}. More precisely, in this model one usually assumes (i) an interview fully reveals the realization of the respective variable, and (ii) the requirement of timely feedback forces the decision maker to irrevocably accept or reject the candidate upon seeing its realization. For i.i.d.\ variables, the best-possible prophet inequality states that a candidate $\sigma$ can be selected such that $\mathbb{E}[X_\sigma]\geq \beta\cdot\mathbb{E}[\max\{X_1,\dots,X_n\}]$, where $\beta\approx 0.745$~\cite{HillK82,CorreaFHOV21}.

The gambler's problem has been extremely popular over the last decades, but assumptions (i) and (ii) are often unrealistic. Even after a long interview, an interviewer is usually not fully aware of the entire set of exact qualifications of a candidate. Moreover, in many selection processes a decision does not have to be taken instantaneously. In this paper, we examine the consequences of an arguably more realistic set of conditions. First, instead of (i), we assume that each candidate only partly reveals their realization in the form of a \emph{single bit of information}. As we observe rather directly, this assumption is asymptotically equivalent to allowing to make a single threshold query to each candidate. Second, instead of (ii), we assume that the selection must be made only at the end of the process.

More formally, we consider the \emph{threshold testing} model. We assume that (i$'$), instead of revelation of $X_i$, we perform a single \emph{threshold test} with some arbitrary threshold $t_i$, and the feedback is binary (``positive'' in case $X_i\geq t_i$ or ``negative'' otherwise), and (ii$'$) a candidate must be chosen only after completing \emph{all} threshold tests. Again denoting the selected candidate by $\sigma$, we are interested in bounding the loss in expected value using an inequality of the form $\mathbb{E}[X_\sigma]\geq c\cdot\mathbb{E}[\max\{X_1,\dots,X_n\}]$ for $c\in(0,1)$. We call this a \emph{semi-online prophet inequality}. A testing algorithm that satisfies it is called \emph{$c$-competitive} and has a \emph{competitive ratio} of $c$.

There are four possible models emerging from the different choices of (i) vs.\ (i$'$) and (ii) vs.\ (ii$'$). The remaining two models do not require substantial analytical efforts. Indeed, when we only replace (i) by (i$'$), the consequences are trivial: There is an optimal algorithm for the standard gambler's problem that uses threshold tests. Thus, the existing optimal algorithm and its guarantees continue to apply because the space of algorithms only shrinks when we require threshold tests. Also, only replacing (ii) by (ii$'$) implies a trivial problem---one can see and choose an option with maximum value at the end, a $1$-competitive strategy. In contrast, the main contribution of this paper is to show that, with (i$'$) and (ii$'$) simultaneously, a mathematically interesting model arises.

Our results also imply a stark qualitative distinction to the standard model. It is well known that adaptivity, i.e., allowing decisions to depend on past observations, does not help for the standard gambler's problem. In our model, we observe rather easily that non-adaptive testing algorithms are unable to asymptotically improve upon the ratio of $\beta\approx 0.745$. Our main result is a set of adaptive algorithms that improves significantly upon this bound and achieves a ratio of approximately $0.869$. To complement this result, we show that there are distributions that imply a non-trivial asymptotical upper bound on the ratio, i.e., there is no $(1-o(1))$-competitive algorithm. We proceed to discuss our contributions in more detail.

\subsection{Techniques and Contribution}

Let $F$ be the cumulative distribution function of the variables. For most of the paper we assume (essentially w.l.o.g.) that $F$ is continuous. Our algorithms perform quantile testing, i.e., they use thresholds of the form $F^{-1}(1-q)$ for $q\in(0,1)$, oblivious of other properties of the distribution. It is straightforward to achieve a competitive ratio of $1-1/\e > 0.632$ by using threshold $t_i = F^{-1}(1-1/n)$ for all variables and then choosing any variable that has been tested positively (if any); see, e.g.,~\cite{HillK82}. The analysis of this strategy is asymptotically tight for each of the following two parametric distributions\footnote{Strictly speaking, $F_A$ and $F_B$ are not continuous. For a rigorous argument, one can resort to an arbitrarily close continuous approximation of the distributions to obtain the same result.}:
\begin{description}
    \item[$F_A$:] For some small $\eps>0$, with probability $1/\sqrt{n}$ choose a value uniformly from $[1-\eps,1+\eps]$, and 0 otherwise. As $\eps\to0$ and $n\to\infty$, the algorithm gets a positive test and therefore value $1$ with probability $1-1/\e$ while $\mathbb{E}[\max\{X_1,\dots,X_n\}]=1$.
    \item[$F_B$:] Choose the value $1$ with probability $1/n^2$ and $0$ otherwise. The algorithm obtains a value 1 with probability $(1-(1-1/n)^n)/n$ while $\max\{X_1,\dots,X_n\}=1$ with probability $1-(1-1/{n^2})^{n}$. As $n\to\infty$, the ratio of both probabilities approaches $1-1/\e$.
\end{description}
To improve upon $1-1/\e$, an algorithm needs to test for \emph{both} smaller and larger thresholds than $F^{-1}(1-1/n)$. Thresholds that are all larger than $F^{-1}(1-1/n)$ decrease the ratio for $F_A$; thresholds that are all smaller than $F^{-1}(1-1/n)$ decrease the ratio for $F_B$.

The class of algorithms we consider here is parameterized by $\alpha_1,\dots,\alpha_k\in(0,1)$ with $\alpha_1>\alpha_2>\dots>\alpha_k$.
In the beginning, such an algorithm uses $F^{-1}(1-\alpha_1/n)$ as a threshold until it sees a positive test. Generally, after $i<k$ positive tests, it sets $F^{-1}(1-\alpha_{i+1}/n)$ as a threshold. After $k$ positive tests (i.e., on $F^{-1}(1-\alpha_1/n),\dots,F^{-1}(1-\alpha_k/n)$), the algorithm can make arbitrary tests. Indeed, we eventually choose $\alpha_1>1$ and $\alpha_2<1$.

In our analysis, we exactly calculate the asymptotic probability that the algorithm sees precisely $i$ positive tests. Note that these probabilities asymptotically determine the probability density function of $X_\sigma$, the chosen variable. It is a step function in quantile space: The probability of making precisely $i$ positive tests is spread uniformly over the interval $[1-\alpha_i/n,1]$ for all $i\in \{1,\dots,k\}$.

We compare this probability density function of $X_\sigma$ with that of $\max\{X_1,\dots,X_n\}$ by stochastic dominance, leading to a tight analysis of such algorithms. For fixed values of $\alpha_1,\dots,\alpha_k$, we can analyze the competitive ratio of the respective strategy by solving a piecewise convex optimization problem, where the $k+1$ pieces correspond to the $k+1$ steps of the step function. We numerically maximize the minimum of this function.

We execute this analysis in detail for $k\in\{2,3\}$. For $k=3$, we obtain a competitive ratio of approximately $0.869$ by setting $\alpha_1\approx 2.035$, $\alpha_2\approx 0.506$, and $\alpha_3\approx 0.057$. Our numerical results for $k=4$ indicate only negligible improvement by further increasing $k$.

We complement this result by a constant upper bound on the competitive ratio, i.e., an impossibility of achieving a competitive ratio of $1-o(1)$. Intuitively, there is a trade-off inherent in every test: Testing for a smaller value yields a fallback option in case only one positive test is found at the end; testing for a larger value allows to differentiate between different variables when multiple of them have been tested positively. There are instances in which, irrespective of how the algorithm solves this trade-off, it loses a constant in the competitive ratio. For the proof we consider a distribution where values 1, 2, or 3 occur with probability $1/n$ each, and $0$ otherwise. It is minimal in the sense that a competitive ratio of $1-o(1)$ is achievable for any distribution that uses only three values in the support, or whose parameters do not depend on~$n$.

Finally, we establish that a competitive ratio of $1-o(1)$ \emph{can} be achieved using $n$ tests when a single variable can be tested multiple times (recall that the realization of each variable is only drawn \emph{once initially} from the distribution). The idea is to drop $o(n)$ variables from consideration and test the remaining ones with a threshold such that, with high probability, $\max\{X_1,\dots,X_n\}$ is larger than this threshold, but only $o(n)$ of these tests are positive. The additional $o(n)$ tests can then be used to find the maximum variable among those that have been tested positively. 

\subsection{Further Related Work}

The original prophet inequality~\cite{KrengelS77} states that there is a $1/2$-competitive algorithm in the setting of independent random variables with arbitrary distributions. Initiated by the work of Hajiaghayi, Kleinberg, and Sandholm~\cite{HajiaghayiKS07}, prophet inequalities have seen a surge of interest in the TCS community over the past 15 years. This has, for instance, led to the development of the tight i.i.d.\ prophet inequality with competitive ratio approximately $0.745$~\cite{CorreaFHOV21} as well as almost-tight random-order~\cite{CorreaSZ21,journals/corr/abs-2211-04145} and free-order~\cite{AgrawalSZ19,LiuLPSS21,PengT22,journals/corr/abs-2211-04145} prophet inequalities. Optimal or near-optimal prophet inequalities can be recovered without knowledge of the distribution but with $O(n)$ samples~\cite{CorreaDFS22,RubinsteinWW20,CorreaDFSZ21,abs-2011-06516}. Several works considered multiple-choice prophet inequalities with combinatorial constraints, e.g.,~\cite{Alaei14,conf/icalp/0002HKSV14,KleinbergW19,DuttingFKL20}. We also refer to the (at this point slightly outdated) surveys of Lucier~\cite{Lucier17} as well as Correa et al.~\cite{CorreaFHOV18} for additional references.

 We compare our work with the two works from the prophet-inequality literature that seem closest. Orthogonally to samples, Li, Wu, and Wu~\cite{journals/corr/abs-2205-05519} considered a version of the unknown i.i.d.\ setting in which quantile queries to the distribution can be made \emph{before} the sequence of variables arrives. They as well as Perez-Salazar, Singh, and Toriello~\cite{journals/corr/abs-2210-05634} also used a limited number of (quantile-based) thresholds to achieve near-optimal i.i.d.\ prophet inequalities. We are not aware of any version of the single-choice prophet inequality with general i.i.d.\ distributions to which the impossibility of approximately $0.745$ does not apply.

In stochastic probing (e.g.,~\cite{AsadpourNS08,GuptaN13,AdamczykSW14,GuptaNS16,GuptaNS17}), information is also revealed online according to known distributions. The standard models are, however, quite different from our model: The decision maker gets to choose which variables to probe, and each probe entirely reveals the realization of the variable at hand. Eventually, the decision maker can pick a (set of) variable(s), much like in our model. Comparing with an omniscient optimum (like in the prophet inequality) is, however, usually hopeless in this setting. Instead, one focusses on computing or approximating the strategy that maximizes the expected selected (total) value, a task that is straightforward for our model. 

In these probing models, the adaptivity gap measures the worst-case multiplicative gap between the value of the best adaptive and that of the best non-adaptive strategy. Note that, while our result does imply a nontrivial adaptivity gap (i.e., larger than 1) for our problem, we are studying a different question as we compare both adaptive and non-adaptive strategies with an omniscient optimum.

We are aware of two works in the probing literature in which tests do not eradicate all uncertainty about the respective variable. Hoefer, Schmand, and Schewior~\cite{conf/ijcai/0001SS21} considered a stochastic-probing model in which the first test to a variable only reveals whether the realization is above or below the median of the distribution, and additional tests can be used to further narrow down the realization in the same way applied to the conditional distribution. Gupta et al.~\cite{GuptaJSS19} generalized the related classic Pandora's box problem due to Weitzman~\cite{Weitzman79} and considered the Markovian model. There, a set of Markov chains, which correspond to variables that can eventually be chosen, is given and, in each step, a probe can be used to advance one of the Markov chains.

Threshold tests have also been considered in the context of estimating (properties of) a probability distribution. For example, Paes Leme et al.~\cite{LemeSTW23} gave bounds on the sample complexity, i.e., required number of such tests, to estimate the approximately optimal reserve price for certain types of distributions. Meister and Nietert~\cite{MeisterN21} as well as Okoroafor et al.~\cite{OkoroaforGK23} investigated the sample complexity of estimating other objects, e.g., mean, median, or even full CDF, of the \emph{empirical} distribution in a non-stochastic setting.

\section{Preliminaries}

We consider \emph{threshold testing} defined as follows. We are given a distribution $F$ on $\mathbb{R}_{\geq 0}$ with finite expectation. There are $n$ \emph{boxes}. Each box $i$ contains a hidden realization $X_1,\ldots,X_n$ drawn \emph{once upfront} i.i.d. from $F$. A \emph{testing algorithm} can apply a \emph{threshold test} to each box $i \in [n] = \{1,\ldots,n\}$ exactly once, in that order. To apply a test to box $i$, the algorithm chooses a threshold $t_i \ge 0$ and receives a binary feedback whether $X_i \geq t_i$ or not. Upon testing $i$, the algorithm learns if $X_i\geq t_i$ or not, but \emph{not the precise value of $X_i$}. If $X_i \ge t_i$ we say the test was \emph{positive}, otherwise it was \emph{negative}. The algorithm may choose thresholds adaptively based on the feedback from earlier tests. Finally, after testing each box, the algorithm chooses one box $\sigma \in [n]$ and receives a reward of $X_{\sigma}$. Here, $\sigma$ is a random variable based on the observed feedback and the internal randomization of the algorithm. We call an algorithm $c$-competitive if $\mathbb{E}[X_\sigma]\geq c\cdot\mathbb{E}[\max\{X_1,\dots,X_n\}]$. We are interested in maxmimizing $c$ in the limit as $n\rightarrow\infty$.

\subparagraph{Non-adaptive Algorithms and Prophet Inequalities.}
Our testing problem has inherent connections to the classic prophet inequality for i.i.d.\ random variables. Consider the \emph{non-adaptive} variant, in which the algorithm chooses thresholds $t_1,\ldots,t_n$ upfront. We observe that this problem is essentially the standard gambler's problem governed by prophet inequalities.
The optimal algorithm for the gambler's problem emerges from straightforward backwards induction. For each box $i \in [n]$, the gambler sets a threshold $t_i$ to the expected profit from the optimal algorithm for boxes $i+1,\ldots,n$. The algorithm accepts $i$ if and only if $X_i \ge t_i$. It is straightforward to verify that this implies $t_1 \ge \ldots \ge t_n$. All $t_i$-values can be computed in advance. As such, a non-adaptive algorithm for threshold testing can use these thresholds and imitate the optimal algorithm for the gambler's problem.
\begin{observation}
    \label{obs:non-adaptive}
    The optimal non-adaptive testing algorithm for $n$ boxes obtains at least the expected reward of the optimal algorithm for the gambler's problem with $n$ boxes.
\end{observation}
We also observe the converse---for large $n$, the optimal reward of non-adaptive threshold testing is essentially the optimal reward in the gambler's problem. 
\begin{proposition}
    \label{prop:non-adaptive}
    The optimal non-adaptive testing algorithm for $n$ boxes obtains at most the expected reward of the optimal algorithm for the gambler's problem with $n+1$ boxes.
\end{proposition}
\begin{proof}
    Consider the optimal non-adaptive algorithm for threshold testing. W.l.o.g.\ we can assume that the chosen thresholds are ordered such that $t_1^* \ge \ldots \ge t_n^*$. If at least one test is positive, then among the positively tested boxes, the algorithm chooses the one with the highest threshold --  which is the earliest one in the sequence. The gambler can imitate this in the online model by using thresholds $t_1^*,\ldots,t_n^*$ and accepting the first one with $X_i \ge t_i^*$. If all tests are negative, then the testing algorithm accepts $X_1$ -- it failed the test with the highest threshold and, as such, has the highest conditional expectation. Clearly, this is less than the apriori expectation of $F$, which can be obtained by the gambler from accepting box $X_{n+1}$. Hence, the gambler with $n+1$ boxes obtains more expected value.
\end{proof}
For large $n$ the best competitive ratio is approximately 0.745 by the optimal prophet inequality~\cite{HillK82,CorreaFHOV21}. For the rest of the paper we focus on \emph{adaptive} testing algorithms.

\subparagraph{Threshold Testing vs.\ General Binary Feedback.}
We discuss our scenario in the context of a more general model. In binary-feedback testing, the algorithm can choose a set $Y_i\subset\mathbb{R}_{\geq 0}$ and learns whether or not $X_i\in Y_i$. Note this model generalizes threshold testing -- setting a threshold $t_i$ can be simulated by choosing $Y_i=\{x\in\mathbb{R}\mid x\geq t_i\}$. Nevertheless, the competitive ratio achievable is asymptotically the same as for threshold testing. As such, we restrict attention to threshold testing. 

\begin{proposition}
    \label{prop:gen-binary}
    The optimal algorithm for binary-feedback testing with $n$ boxes obtains at most the expected reward of the optimal algorithm for threshold testing with $n+1$ boxes.
\end{proposition}

\begin{proof}
    Consider an optimal algorithm for binary-feedback testing with $n$ boxes. We modify this algorithm to obtain an algorithm for threshold testing with $n+1$ boxes. We assume w.l.o.g.\ that, whenever the original algorithm chooses a set $Y_i$ to test the $i$-th box, then $\mathbb{E}[X_i \mid X_i\in Y_i]\geq\mathbb{E}[X_i]$ and, therefore, $\mathbb{E}[X_i \mid X_i\notin Y_i]\leq \mathbb{E}[X_i]$. We replace any such test with a threshold test for a threshold $t_i$ such that $\Pr[X_i\geq t_i]=\Pr[X_i\in Y_i]$, i.e., both tests are positive with precisely the same probability and $\mathbb{E}[X_i \mid X_i\geq t_i]\geq \mathbb{E}[X_i \mid X_i\in Y_i]$. We continue in the same way as the original algorithm would upon a positive or negative test, modifying subsequent tests in the same way. If the original algorithm eventually picks a box $i^\star$ with a positive test result, the new algorithm picks the same box. Thereby it obtain at least the same value since, by our choice of $t_{i^\star}$, $\mathbb{E}[X_{i^\star}\mid X_{i^\star}\geq t_{i^\star}]\geq \mathbb{E}[X_{i^\star}\mid X_{i^\star}\in Y_{i^\star}]$. Similarly, if the original algorithm would pick a box with a negative result, the new algorithm picks box $n+1$, obtaining $\mathbb{E}[X_{n+1}]=\mathbb{E}[X_{i^\star}]\geq \mathbb{E}[X_{i^\star}\mid X_{i^\star}\notin Y_{i^\star}]$ by our assumption above.
\end{proof}
\section{Adaptive Testing}
\label{sec:algo}

In this section, we prove the following theorem. For simplicity, we consider a continuous distribution $F$ throughout the proof. In the following section, we discuss that the result also generalizes to finite discrete distributions. 

\begin{theorem}\label{thm:iid-alg}
	There is an efficient $(0.869-o(1))$-competitive algorithm for threshold testing with a continuous distribution.
\end{theorem}

\begin{proof}
We consider a class of algorithms that is parameterized by a monotone sequence of quantile parameters $q_1,\dots,q_k \in (0,1)$ where $q_1 > \ldots > q_k$. For convenience, we assume $q_0 = 1$ and $q_{k+1} = \ldots = q_n = 0$. The algorithm starts by testing for the $1-q_1$ quantile of $F$. Since the distribution is continuous, $q_1$ corresponds to a threshold $\tau_1$ (i.e., $\tau_1$ is such that $Pr[X_i \ge \tau_1] = q_1$). Then for any $j \ge 1$, if the algorithm sees a negative test for $\tau_j$, it continues testing with $\tau_j$. If it sees a positive test for $\tau_j$, it increments $j$ to $j+1$ (i.e., continues testing with the next threshold $\tau_{j+1}$). After having tested each box, it selects the one with the best conditional expectation. This is either the box tested positively for the threshold corresponding to the largest quantile, or any box (when all tested negative for $\tau_1$).

We consider the values of $q_j$ in the form $q_j = \alpha_j/n$ for some $\alpha_j\in (0,n)$, for all $j\in[k]$. In Table~\ref{tab:iid-results}, we give example values of $\alpha_j$ and the resulting competitive ratios for different values of $k$. We obtained these values by numerical optimization over a bounded interval.
\begin{table}[t]
	\centering
	\begin{tabular}{cccccl}
		\toprule
		$k$ & $\alpha_1$ & $\alpha_2$ & $\alpha_3$ & $\alpha_4$ & comp.\ ratio as $n\rightarrow\infty$\\
		\midrule
		$1$ & $1$ & -- & -- & -- & $1-1/e\approx0.63212$\\
		$2$ & $1.83298$ & $0.35932$ & -- & -- & $>0.84005$\\
		$3$ & $2.035135$ & $0.5063$ & $0.05701$ & -- & $>0.86933$\\
		$4$ & $2.038$ & $0.508$ & $0.058$ & $0.0002$ & $>0.86956$\\					
		\bottomrule
	\end{tabular}
	\caption{Numerically optimized parameters and competitive ratios for different values of $k$.}
	\label{tab:iid-results}
\end{table}

We use $F$ to denote the CDF, i.e., $F(x) = \Pr[X_i < x]$, for each $i \in [n]$ and $x \in [0,1]$. For the maximum over $n$ i.i.d.\ draws, we obtain the CDF $F_m(x) = (F(x))^n = (\Pr[X_i < x])^n = \prod_i (\Pr[X_i < x]) = \Pr[\max_i X_i < x]$. We denote the complementary CDF by $G(x) = \Pr[X_i \ge x] = 1 - F(x)$ and $G_m(x) = \Pr[\max_i X_i \ge x] = 1 - F_m(x)$. Since $F$ is continuous, threshold $\tau_j = G^{-1}(q_j) = F^{-1}(1-q_j)$, i.e., $G(\tau_j) = q_j$ and $F(\tau_j) = 1-q_j$. Similarly, $F_m(\tau_j) = (1-q_j)^n$ and $G_m(\tau_j) = 1 - (1-q_j)^n$. We here restrict attention to values of $\alpha_j \in o(n)$, we will assume these are constants throughout. This implies $\lim_{n \to \infty} G_m(\tau_j) = 1 - \e^{-\alpha_j}$.

Our analysis proceeds via stochastic dominance. For any given threshold $t \ge 0$ we compare the complementary CDF $G_m(t)$ to the complementary CDF of our algorithm. We denote the latter by $A(t) = \Pr[X_{\sigma} \ge t]$, where $\sigma$ is the box chosen by our algorithm. If $A(t) \ge c \cdot G_m(t)$ for all $t\ge 0$, then the algorithm is $c$-competitive by stochastic dominance.

For any given $t \in [0,\infty)$ let $q = G(t) = 1-F(t)$ and $\alpha = n\cdot q$. We will conduct our analysis with respect to $\alpha \in [0,n]$ instead of $t \in [0,\infty)$. We split $[0,n]$ into intervals $I_j = [\alpha_{j+1},\alpha_j]$ for $j=0,\ldots,k$, where we use $\alpha_0 = n$ and $\alpha_{k+1} = 0$. Suppose we see a positive test for $\alpha_j$. Then, between the positive test for $\alpha_{j-1}$ and the one for $\alpha_j$, assume there are $\ell_j \ge 0$ negative tests.

\subparagraph{Two Thresholds.}
We start by discussing an algorithm with $k=2$ thresholds. Suppose $\alpha \in I_2$. First, let's assume we only have a positive test for $t_1$ but not for $t_2$. We call this event $\calE_{10}$. It happens with probability
\begin{align*}
\Pr[\calE_{10}] &= \sum_{\ell_1=0}^{n-1} (1-q_1)^{\ell_1} q_1 \cdot (1-q_2)^{n-1-\ell_1} = q_1 \cdot (1-q_2)^{n-1} \cdot \frac{1-\left(\frac{1-q_1}{1-q_2}\right)^n}{1-\frac{1-q_1}{1-q_2}} \\
&= q_1 \cdot \frac{(1-q_2)^n - (1-q_1)^n}{q_1 - q_2} \quad = \quad \alpha_1 \cdot \frac{\left(1 - \frac{\alpha_2}{n}\right)^n - \left(1-\frac{\alpha_1}{n}\right)^n}{\alpha_1 - \alpha_2}\enspace.
\end{align*}
In this case, the algorithm selects the box that was tested positive for $\tau_1$. It has a value at least $t$ with probability $q/q_1 = \alpha/\alpha_1$.

Otherwise, we have a positive test for $\tau_1$ and $\tau_2$, which we call event $\calE_{11}$. The event that we have a positive test for $\tau_1$ (irrespective of what happens for $\tau_2$) is called $\calE_1$. Clearly,
\begin{align*}
\Pr[\calE_{11}] &= \Pr[\calE_1] - \Pr[\calE_{10}]. 
\end{align*}
In case $\calE_{11}$ happens, we select the box that tested positive for $\tau_2$. It has a value at least $t$ with probability $q/q_2 = \alpha/\alpha_2$.

Overall, for $\alpha \in I_2$ we see
\begin{align*}
A(\alpha) &= \frac{\alpha}{\alpha_1} \Pr[\calE_{10}] + \frac{\alpha}{\alpha_2} \Pr[\calE_{11}] = \alpha\cdot \left(\frac{\Pr[\calE_{10}]}{\alpha_1} + \frac{\Pr[\calE_1] - \Pr[\calE_{10}]}{\alpha_2}\right)\\
 &= \alpha \cdot \left(\frac{\Pr[\calE_1]}{\alpha_2}  - \frac{(\alpha_1 - \alpha_2)\Pr[\calE_{10}]}{\alpha_1\alpha_2}\right) = \frac{\alpha}{\alpha_2} \cdot \left(1 
 - \left(1-\frac{\alpha_2}{n}\right)^n\right) \\
 &= c_2(\alpha) \cdot \left(1-\left(1-\frac{\alpha}{n}\right)^n\right) = c_2(\alpha) \cdot G_m(\alpha)
\end{align*}
Hence,
\[
    c_2(\alpha) = \frac{\alpha}{\alpha_2} \cdot \frac{1-\left(1-\frac{\alpha_2}{n}\right)^n}{1-\left(1-\frac{\alpha}{n}\right)^n} \ge \lim_{\alpha \to 0} c_2(\alpha) = \frac{1-\left(1-\frac{\alpha_2}{n}\right)}{\alpha_2} \ge \frac{1-\e^{-\alpha_2}}{\alpha_2}\enspace,
\]
since for every given $n \ge 1$ and every $\alpha > 0$, the ratio $\alpha/(1-(1-\alpha/n)^n) > 1$, because $\alpha \ge 1-(1-\alpha/n)^n$ by concavity of the latter function. 

Now for $\alpha \in I_1$, we consider the case with a positive test on $\tau_1$ but not on $\tau_2$. In this case, the box has a value of at least $t$ with probability $q/q_1 = \alpha/\alpha_1$. Alternatively, if we see a positive test for $\tau_1$ and $\tau_2$, the algorithms selects a box with a value of at least $t$ with probability 1. Overall, for $\alpha \in I_1$
\begin{align*}
    A(\alpha) &= \frac{\alpha}{\alpha_1} \cdot \Pr[\calE_{10}] + \Pr[\calE_{11}] \; = \; \Pr[\calE_1] - \left(\frac{\alpha_1 - \alpha}{\alpha_1}\right)\cdot \Pr[\calE_{10}]\\
    &= 1-\left(1- \frac{\alpha_1-\alpha}{\alpha_1-\alpha_2}\right) \left(1-\frac{\alpha_1}{n}\right)^n - \frac{\alpha_1-\alpha}{\alpha_1 - \alpha_2} \left(1-\frac{\alpha_2}{n}\right)^n \\
    &= c_1(\alpha) \cdot \left(1-\left(1-\frac{\alpha}{n}\right)^n\right),
\end{align*}
Since $\alpha \in [\alpha_2,\alpha_1]$ is a constant,
\[
    \lim_{n \to \infty} c_1(\alpha) = 
    \frac{1}{1-\e^{-\alpha}} \cdot \left(1 - \e^{-\alpha_1} - \frac{(\alpha_1 - \alpha) (\e^{-\alpha_2} - \e^{-\alpha_1})}{\alpha_1 - \alpha_2}\right).
\]

Finally, for $\alpha \in I_0$, we see that
\begin{align*}
A(\alpha) = \frac{\alpha - \alpha_1}{n-\alpha_1} (1-\Pr[\calE_1]) + \Pr[\calE_1] 
          &= \frac{\alpha - \alpha_1}{n-\alpha_1} \left(1-\frac{\alpha_1}{n}\right)^{n} + \left(1-\left(1-\frac{\alpha_1}{n}\right)^n\right) \\
          &= c_0(\alpha)\cdot \left(1-\left(1-\frac{\alpha}{n}\right)^n\right)
\end{align*}
Thus,
\begin{align*}
c_0(\alpha) \ge \frac{1-\left(1-\frac{\alpha_1}{n}\right)^n}{1-\left(1-\frac{\alpha}{n}\right)^n} \ge \frac{1-\e^{-\alpha_1}}{1-\e^{-\alpha}}
\end{align*}
where the latter bound holds for any $n \ge 1$ and any constant $\alpha$. Indeed, when $\alpha \in \omega(1)$, we obtain a bound of $1-\e^{-\alpha}$ in the limit for $n \to \infty$.

As a sanity check, observe that $c_1(\alpha_1) = c_0(\alpha_1) = 1$. Indeed, suppose we have a box with value $t \ge \tau_1$. Then either this box is tested positive for $\tau_1$, or some other box was tested positive for $\tau_1$ before. In either case, the algorithm indeed selects a box of value at least $\tau_1$. Similarly, observe that $c_2(\alpha_2) = 1$ as well. Indeed, suppose we have a box with value $t \ge \tau_2$. Suppose (1) this box is tested positive for $\tau_2$. Then it is selected. Suppose (2) the box is tested positive for $\tau_1$. Then it is selected, unless some later box is tested positive for $\tau_2$. Either way, we eventually obtain a value of at least $\tau_2$. Finally, suppose (3) the box is not tested at all. Then we have already selected a box of value at least $\tau_2$ before.

To obtain the best ratio, we strive to select constants $0 < \alpha_2 < \alpha_1$ in order to
\[
    \max_{\alpha_1,\alpha_2} \; \{ \min_{\alpha \in I_2} c_2(\alpha), \, \min_{\alpha \in I_1} c_1(\alpha), \, \min_{\alpha \in I_0} c_0(\alpha) \}
\]

For $c_2(\alpha)$ and $c_0(\alpha)$ we obtain fairly clear lower bounds, which even hold pointwise for any $n$. It seems unpromising to obtain an insightful analytic formula for the minimum of $c_1(\alpha)$ as a function of $\alpha_1$ and $\alpha_2$. Instead, we numerically optimized parameters $\alpha_1,\alpha_2$ and used standard solver software to minimize $c_1(\alpha)$. The lower bounds for $c_2$ and $c_0$ then amount to $(1-\e^{-0.35932})/0.35932 \ge 0.8400637\ldots$ and $1-\e^{1.83298} \ge 0.8400564\ldots$. The minimum of $\lim_{n \to \infty} c_1(\alpha)$ is located roughly at $\alpha^* \approx 0.832961265\ldots$ with a value for $c_1(t) = 0.8400569\ldots$ For a plot of the ratios see Figure~\ref{fig:2plot}.
\input{fig_2plot}

Along similar lines, we analyze the case with $k=3$ thresholds in the appendix, which yields a ratio of at least $0.869-o(1)$ (see Table~\ref{tab:iid-results}). Based on similar calculations, we also numerically optimized the case with $k=4$, but we see only very slight improvements. Intuitively, the probability to reach a state with positive tests for all $k$ thresholds becomes extremely small. Increasing this probability requires to decrease the value to be tested for in the first $k-1$ tests. However, the possibility to obtain an improvement in this way seems to vanish very quickly as $k$ grows larger. We conjecture that for all values of $k$, we cannot significantly improve the competitive ratio beyond 0.869 as $n \to \infty$.
\end{proof}

Observe that the analysis of our algorithms is tight. Consider the value of $\alpha'$ that yields the minimum of all $c_i(\alpha)$ in the respective intervals $I_i$. For a ``golden nugget''-distribution, where each $X_i$ has value 1 with probability $\alpha'/n$ and 0 otherwise, the above calculations become exact, and the analysis of the competitive ratio becomes tight. While, strictly speaking, this golden-nugget distribution is discrete, it is straightfoward to approximate it arbitrarily closely by a continuous distribution.

\section{Discrete Distributions}

\label{sec:discrete}

Let us shift attention from a continuous distribution to a finite discrete distribution $F$. We assume $F$ is represented in straightforward form as a list of (value, probability) pairs. We denote by $m$ the number of distinct realizations, and we use $v_1 < v_2 < \ldots < v_m$ to denote the support of $F$.

Observe that w.l.o.g.\ we only need to test for these values $v_j$. If we test for a threshold $t$ in between two consecutive $v_j < t \le v_{j+1}$, we obtain the same result by testing for $t = v_{j+1}$ instead. As such, we restrict to tests for values in the support.

\subsection{Testing Algorithms}
\label{sec:optimal}

We first observe that an optimal testing algorithm can be computed in polynomial time. Moreover, we show that this algorithm yields a competitive ratio of $0.869-o(1)$.

\begin{theorem}
    \label{thm:dynprog}
    For finite discrete distributions, an optimal testing algorithm can be computed by dynamic programming in polynomial time.
\end{theorem}

\begin{proof}
    We use backwards induction. Consider the last test of box $n$. Clearly, given the previously tested boxes $1,\ldots,n-1$, we can restrict attention to the one with the highest conditional expectation. We denote this value by $V_{n-1}^*$. Since each box is tested for exactly one of the $m$ realizations, there are $2m$ different possibilities for $V_{n-1}^*$. There are $m$ possible tests for box $n$. We can enumerate all the $2m^2$ combinations. For each value of $V_{n-1}^*$, the optimal test of box $n$ is the one leads to the highest expected value of the chosen item. Thus, to determine and describe the optimal decision for box $n$, we only need to consider $2m$ options of $V_{n-1}^*$, and for each option we record the best of the $m$ possible tests for box $n$.

    For the induction, let $V_{i-1}^*$ and $V_i^*$ be the conditional expectation of the best tested box before and after testing box $i$, resp. Suppose that for each possible value of $V_i^*$, we have computed an optimal testing strategy for subsequent boxes $i+1,\ldots,n$, along with the resulting expected value of $X_{\sigma}$. Now for box $i$, consider each of the $2m$ possible values for $V_{i-1}^*$. For each realization $v_k$, we can determine the effect when we test box $i$ for $v_k$ -- i.e., the probability that $V_i^* = V_{i-1}^*$ (when the test on $i$ implies the conditional expectation of $i$ is at most $V_{i-1}^*$), as well as the probability that $V_i^*$ becomes any higher value (otherwise). For the resulting $V_i^*$, we inspect the value obtained by an optimal testing strategy for boxes $i+1,\ldots,n$. This serves to find the test of box $i$ resulting in the optimal expected value.
 
    Overall, to determine and describe the optimal decision for box $i$, we need to consider $2m$ options of $V_{i-1}^*$, and for each option we determine the best of the $m$ possible tests for box $i$ (using the results of the subsequent optimal testing strategy for boxes $i+1,\ldots,n$). Finally, for box 1 $V_0^*$ is undefined. At this point, we only need to find the best of the $m$ possible tests for box 1 (using the results of the subsequent optimal testing strategy for boxes $2,\ldots,n$). This concludes the backwards induction.

    We record for each possible value $V_{i-1}^*$ the best threshold to test box $i$ along with the resulting expected value emerging from an optimal algorithm for boxes $i+1,\ldots,n$. Hence, we can describe an optimal testing strategy using $2 \cdot (1 + (n-1)\cdot2m))$ entries. This strategy can be computed in polynomial time via dynamic programming as described above.
\end{proof}

At this point, it is unclear how to apply our algorithm from the previous section to finite discrete distributions since $F^{-1}(1-q)$ may not be defined for the relevant values of $q$. In fact, we will show that the optimal algorithm in Theorem~\ref{thm:dynprog} achieves a competitive ratio of at least $0.869 - o(1)$ for finite discrete distributions.

We consider the following different model for testing discrete distributions, called \emph{probability testing}. It can be viewed as the limit that emerges from approximating discrete with continuous distributions arbitrarily close. Here a test requires an input value $q \in [0,1]$. It then returns whether or not the value $v$ in the box lies in the top-$q$ fraction of probability mass of $F$. For a finite discrete distribution $F$, let $k$ be such that $\sum_{j=k+1}^m \Pr[v = v_j] < q \le \sum_{j=k}^{m} \Pr[v = v_j]$. Then the test is positive for $v \in \{v_{k+1},\ldots,v_m\}$ and negative for $v \in \{v_1,\ldots,v_{k-1}\}$. For $v = v_k$ it yields a randomized outcome, i.e., positive with probability $p_q = \left(q - \sum_{j=k+1}^m \Pr[v = v_j]\right)/\Pr[v = v_k]$ and negative otherwise. Hence, the overall probability that box $i$ is tested positive is exactly $q$.

Clearly, our algorithm from Section~\ref{sec:algo} can be implemented with probability testing and obtains a competitive ratio of $0.8969-o(1)$. Probability and threshold testing are equivalent for continuous distributions, since there is a bijection between thresholds and values for $q$. For finite discrete distributions we observe in Proposition~\ref{prop:opt-real} that any algorithm for probability testing can be simulated using randomized threshold tests. We then show that randomized tests are not beneficial, i.e., for any algorithm with randomized threshold tests, there is one with \emph{deterministic} tests performing at least as good.

\begin{proposition}\label{prop:opt-real}
	If there is a $c$-competitive algorithm for probability testing, then there is a $c$-competitive algorithm for threshold testing.
\end{proposition}
\begin{proof}
Consider a probability test with parameter $q$ and let $k$ be such that $\sum_{j=k+1}^m \Pr[v = v_j] < q \le \sum_{j=k}^{m} \Pr[v = v_j]$. The outcome of the test can be simulated by executing a randomized threshold test: We test for $v_k$ with probability $p_q$ and for $v_{k+1}$ otherwise. Then the randomized threshold test is positive for $v \in \{v_{k+1},\ldots,v_m\}$ and negative for $v \in \{v_1,\ldots,v_{k-1}\}$. For $v = v_k$ it is positive with probability $p_q$. Hence, it represents a valid simulation of the probability test.

Consider any algorithm $A_p$ for probability testing. It can be interpreted as a randomized algorithm $A^r_t$ for threshold testing. Formally, when $A_p$ uses a probability test on box $i$, $A^r_t$ uses the equivalent randomized threshold test. $A_p$ learns the outcome of the test but not the information which of the two threshold tests was chosen to generate the result. Clearly, by ignoring the latter information, we can generate the same behavior using $A^r_t$.

Let us describe an iterative adjustment to obtain a \emph{deterministic} algorithm $A^d_t$ that achieves a better performance than $A^r_t$ (and, thus, $A_p$). First, suppose we have tested all the boxes. Then $A^r_t$ chooses the box with the highest conditional expectation. Now suppose $A^d_t$ behaves exactly like $A^r_t$, but right before choosing the box it is allowed to see the information about which threshold test was actually applied to each box. This can only lead to a better choice. Thus, $A^d_t$ gets to see the threshold for which box $n$ was tested.

We now apply a similar argument using backwards induction. Suppose $A^d_t$ tests exactly like $A^r_t$ until box $i-1$. Consider the randomized threshold test applied to box $i$ by $A^r_t$ with value $q$. When applying this randomized test, $A^d_t$ sees the threshold for which $i$ was tested ($v_k$ or $v_{k+1}$). By the inductive hypothesis, we assume that there is an optimal testing algorithm using only deterministic threshold tests for subsequent boxes $i+1,\ldots,n$.

Clearly, the value of a randomized test merely represents a convex combination of the values obtained by testing either $v_k$ or $v_{k+1}$, along with the optimal testing of subsequent boxes. As such, $A^d_t$ can only improve by choosing the better of the two tests deterministically. It obtains a deterministic testing strategy for boxes $i,\ldots,n$ that is at least as good as that of $A^r_t$. Consequently, there is an optimal algorithm for boxes $i,\ldots,n$ that uses deterministic threshold tests. This proves the inductive hypothesis.

Overall, this shows that for any algorithm $A_p$ for probability testing, we can obtain an equivalent algorithm $A^r_t$ for randomized threshold testing, whose performance is dominated by a deterministic algorithm for threshold testing.
\end{proof}

\begin{corollary}
    The optimal algorithm for finite discrete distributions is at least $(0.869-o(1))$-competitive for threshold testing.
\end{corollary}

\subsection{Impossibility}
Complementing our results in the previous subsection, we proceed to show a constant upper bound on the competitive ratio for $n \to \infty$.

\begin{theorem}
    \label{thm:impossible}
	There exists no $(1-o(1))$-competitive algorithm for threshold testing.
\end{theorem}

To prove the theorem, we are going to construct a counter example that is a discrete distribution, which carries over to the continuous case by the arguments given in Section~\ref{sec:discrete}. We first observe that such a distribution needs to depend on $n$: Otherwise, the top realization appears with constant probability in each box, and an algorithm simply testing for that realization finds it with probability $1-o(1)$. Furthermore, such a distribution needs to have a support of cardinality at least $4$: If the cardinality of the support is $3$, it is w.l.o.g.\ exactly $3$, and the algorithm can obtain $\max\{X_1,\dots,X_n\}$ by testing for the middle realization and, upon a positive test, testing for the top realization. If it finds a positive test on the top realization, it clearly obtains $\max\{X_1,\dots,X_n\}$ by choosing the corresponding box. If it finds a positive test on the middle realization, the corresponding box is the only one that can possibly contain the top realization, which the algorithm obtains by picking it, so it also obtains $\max\{X_1,\dots,X_n\}$. In the final case, $\max\{X_1,\dots,X_n\}$ is only the lowest realization, which the algorithm will also obtain by choosing any box.

We consider boxes that contain a realization $3$, $2$, or $1$ with probability $1/n$ each and $0$ otherwise. Intuitively, any algorithm that does not always test for the value $1$ before encountering a positive test runs the risk of missing a value $1$. Similarly, any algorithm that does not always test for the value $2$ afterwards and before encountering another positive test runs the risk of missing a value $2$. Such an algorithm, however, with constant probability, gets into a situation in which it has encountered precisely two positive tests, specifically, for the values $1$ and $2$. In that situation, it is clearly optimal to choose the box that has been positively tested for the value $2$. With a constant probability, the value of this box is, however, equal to $2$ while the one that has been tested positively for value $1$ is equal to $3$. The conclusion is that the algorithm, in any case, loses a constant fraction of $\mathbb{E}[\max\{X_1,\dots,X_n\}]$. In the following, we present a formal version of this argument.

\begin{proof}[Proof of Theorem~\ref{thm:impossible}.]
    Let $n\geq 3$. For each $i\in[n]$, let
    \begin{equation*}
        X_i =
        \begin{cases}
            3 & \text{w.p. }1/n,\\
            2 & \text{w.p. }1/n,\\
            1 & \text{w.p. }1/n,\\
            0 & \text{w.p. }1-3/n.
        \end{cases}
    \end{equation*}
    First note that
    \begin{align*}
        &\mathbb{E}[\max\{X_1,\dots,X_n\}] \\
        =\;& \left(1-\left(1-\frac1n\right)^n\right)\cdot 3 + \left(\left(1-\frac1n\right)^n-\left(1-\frac2n\right)^n\right)\cdot 2 +\left(\left(1-\frac2n\right)^n-\left(1-\frac3n\right)^n\right)\cdot 1\\
        =\;& \left(1-\frac1\e\right)\cdot 3 + \left(\frac1\e-\frac1{\e^2}\right)\cdot 2 + \left(\frac1{\e^2}-\frac1{\e^3}\right)\cdot 1+o(1)=\Theta(1).
    \end{align*}
    
    Now consider an arbitrary algorithm \alg. We denote by $\tau$ the possibly random index of the box that it chooses. As discussed above, we can assume w.l.o.g.\ that \alg only tests for the realizations $3$, $2$, and $1$. We first consider the initial phase in which \alg has not yet seen any positive test. Let $T$ be the (random) set of boxes tested in the initial phase, and thereof let $T_i$ be the boxes tested on value $i\in\{1,2,3\}$. Then 
    \begin{equation}
        |T| = |T_1| + |T_2| + |T_3|. \label{eq:lb-T-sum}
    \end{equation}
    Also note that we have
    \begin{equation}
        \mathbb{E}[|T|]\geq \sum_{i=1}^n \left(1-\frac3n\right)^{i-1}\geq \left(\frac{\e^3-1}{3\e^3}\right)\cdot n \label{eq:lb-lb-exp-T}
    \end{equation}
    because each test is positive with probability at most $3/n$.
    
    We first consider two cases with respect to the expected cardinality of $T_1$.
    \begin{enumerate}[itemindent=1cm]
        \item[Case 1:] It holds that $\mathbb{E}[|T_1|]\leq\frac{\e^3-1}{6\e^3}\cdot n$. Note that combining this with Eq.~\eqref{eq:lb-T-sum}, Eq.~\eqref{eq:lb-lb-exp-T}, and $|T|\leq n$  yields
        \begin{equation}
            \Pr\left[|T_2|+|T_3|\geq\frac{\e^3-1}{12\e^3}\cdot n\right]\geq\frac{\e^3-1}{1+11\e^3}. \label{eq:lb-prob-e2}
        \end{equation}
        
        We define several events:
        \begin{itemize}
            \item $\calE_1$ is the event that \alg only performs negative tests,
            \item $\calE_2$ is the event that indeed $|T_2|+|T_3|\geq\frac{\e^3-1}{12\e^3}\cdot n$,
            \item $\calE_3$ is the event that $\max\{X_1,\dots,X_n\}\geq 1$.
        \end{itemize}
        We denote by $\calE^\star$ the intersection of these three events. Note that
        \begin{align*}
            \Pr[\calE^\star]=&\;\Pr[\calE_1]\cdot\Pr[\calE_2\mid \calE_1]\cdot\Pr[\calE_3\mid \calE_1\cap\calE_2]\\
            >&\left(1-\frac3n\right)^n\cdot \frac{\e^3-1}{1+11\e^3}\cdot \left(1-\frac1n\right)^{\frac{\e^3-1}{12\e^3}\cdot n}\geq\delta_1,
        \end{align*}
        for some constant $\delta_1>0$, where we bound all the probabilities separately for the first inequality. For the first probability, we use that each test is negative with probability at least $1-3/n$ independently. For the second probability, we use Eq.~\eqref{eq:lb-prob-e2} and that $|T_1|$, $|T_2|$, and $|T_3|$ are largest when conditioned on $\calE_1$. Finally, for the third probability, we use that, conditioned on $\calE_1\cap\calE_2$, the probability that any variable in $T_2\cup T_3$ is nonzero with probability larger than $1/n$ independently.
        Observing that
        \begin{align*}
            \mathbb{E}[X_\sigma]&=\Pr[\overline{\calE^\star}]\cdot\mathbb{E}[X_\sigma\mid\overline{\calE^\star}]+\Pr[\calE^\star]\cdot\mathbb{E}[X_\sigma\mid\calE^\star]=\Pr[\overline{\calE^\star}]\cdot\mathbb{E}[X_\sigma\mid\overline{\calE^\star}]+o(1)
        \end{align*}
        while
        \begin{align*}
            &\mathbb{E}[\max\{X_1,\dots,X_n\}]\\
            =\;&\Pr[\overline{\calE^\star}]\cdot\mathbb{E}[\max\{X_1,\dots,X_n\}\mid\overline{\calE^\star}]+\Pr[\calE^\star]\cdot\mathbb{E}[\max\{X_1,\dots,X_n\}\mid\calE^\star]\\
            \geq\;&\Pr[\overline{\calE^\star}]\cdot\mathbb{E}[\max\{X_1,\dots,X_n\}\mid\overline{\calE^\star}]+\Pr[\calE^\star]\cdot 1,
        \end{align*}
        we apply $X_\sigma\leq\max\{X_1,\dots,X_n\}$ and obtain
        \begin{align*}
            \mathbb{E}[X_\sigma]\leq (1-\delta_1 + o(1))\cdot\mathbb{E}[\max\{X_1,\dots,X_n\}].
        \end{align*}
        \item[Case 2:] It holds that $\mathbb{E}[|T_1|]>\frac{\e^3-1}{6\e^3}\cdot n$. Recall that any test for the value $1$ is positive with probability $3/n$. Therefore, with probability at least $\frac{\e^3-1}{6\e^3}\cdot n\cdot \frac3n=\frac{\e^3-1}{2\e^3}$, \alg gets into a situation in which it has performed a single positive test and that test was a test on value $1$. Let $U$ be the (random) set of boxes tested by \alg in this situation (if \alg does not get into this situation, $U$ is empty). Similarly to the previous case, let $U_i$ be the set of these boxes tesed for value $i\in\{2,3\}$. We have
        \begin{equation}
        	|U| = |U_2| + |U_3| \label{eq:lb-T-sum-2}
        \end{equation}
        and
        \begin{align}
        	\mathbb{E}[|U|] \geq \frac{\e^3-1}{4\e^3}\cdot\sum_{i=1}^{\big\lfloor\frac{\e^3-1}{12\e^3}\cdot n\big\rfloor} \left(1-\frac{2}{n}\right)^{i-1}\;&\geq \left(\frac{(1-\e^3)(\e^{\frac{1-\e^3}{6\e^3}}-1)}{8\e^3}-o(1)\right)\cdot n \nonumber\\
        	&\geq \left(\frac{1}{100}-o(1)\right)\cdot n. \label{eq:lb-lb-exp-T-2}
        \end{align}
        For~\eqref{eq:lb-lb-exp-T-2}, we use that, with probability at least $\frac{\e^3-1}{4\e^3}$, \alg gets into the aforementioned situation no later than round $\lceil(1-\frac{\e^3-1}{12\e^3})\cdot n\rceil$ since at most a probability of $3/n$ can stem from each round. Further, any test performed in that situation is successful with probability at most $2/n$.
        
        We consider two cases with respect to the expected cardinality of $U_2$.
		\begin{enumerate}[itemindent=1cm]
			\item[Case 2a:] It holds that $\mathbb{E}[|U_2|]\leq\frac{n}{200}$. Similarly as in Case 1, by combining $|U|\leq n$ with Eq.~\eqref{eq:lb-T-sum-2} and Eq.~\eqref{eq:lb-lb-exp-T-2}, we get
	    \begin{equation}
            \Pr\left[|U_3|\geq \frac{n}{400}\right]\geq \frac{1}{399}-o(1). \label{eq:lb-prob-f2}
        \end{equation}
        
        We define several events:
        \begin{itemize}
            \item $\mathcal{F}_1$ is the event that \alg only performs a single positive test, namely on value $1$,
            \item $\mathcal{F}_2$ is the event that, if $\mathcal{F}_1$ occurs, then the realization of the corresponding box is $1$,
            \item $\mathcal{F}_3$ is the event that indeed $|U_3|\geq \frac{n}{400}$,
            \item $\mathcal{F}_4$ is the event that $\max\{X_1,\dots,X_n\}\geq 2$.
        \end{itemize}
        We denote by $\mathcal{F}^\star$ the intersection of these three events. Note that
        \begin{align*}
            \Pr[\mathcal{F}^\star]=&\;\Pr[\mathcal{F}_1]\cdot\Pr[\mathcal{F}_2\mid \mathcal{F}_1]\cdot\Pr[\mathcal{F}_3\mid \mathcal{F}_1\cap\mathcal{F}_2]\cdot\Pr[\mathcal{F}_4\mid \mathcal{F}_1\cap\mathcal{F}_2\cap\mathcal{F}_3]\\
            >&\frac{\e^3-1}{2\e^3}\cdot\left(1-\frac3n\right)^n\cdot\frac13\cdot\left(\frac{1}{399}-o(1)\right)\cdot \left(1-\frac1n\right)^{\frac{n}{400}}\geq\delta_2-o(1),
        \end{align*}
        for some constant $\delta_2>0$, where we bound all the probabilities separately for the first inequality. For the first probability, we use that the algorithm gets into the above situation with probability at least $\frac{\e^3-1}{2\e^3}$, and each of the less than $n$ subsequent tests is negative with probability at least $1-3/n$ independently. For the second probability, we observe that the value of any of the boxes, conditioned on it being at least $1$, is equal to $1$ with probability $1/3$. For the third probability, we use Inequality~\eqref{eq:lb-prob-f2} and that $|U_2|$ and $|U_3|$ are largest when conditioned on $\mathcal{F}_1$ (their sizes are independent of $\mathcal{F}_2$). Finally, for the fourth probability, we use that, conditioned on $\mathcal{F}_1\cap\mathcal{F}_2\cap\mathcal{F}_3$, any box in $U_3$ has value $2$ with probability larger than $1/n$ independently.
        Observing that
        \begin{align*}
            \mathbb{E}[X_\sigma]&=\Pr[\overline{\mathcal{F}^\star}]\cdot\mathbb{E}[X_\sigma\mid\overline{\mathcal{F}^\star}]+\Pr[\mathcal{F}^\star]\cdot\mathbb{E}[X_\sigma\mid\calE^\star]\\
            &=\Pr[\overline{\mathcal{F}^\star}]\cdot\mathbb{E}[X_\sigma\mid\overline{\mathcal{F}^\star}]+\Pr[\mathcal{F}^\star]\cdot 1
        \end{align*}
        while
        \begin{align*}
            &\mathbb{E}[\max\{X_1,\dots,X_n\}]\\
            =\;&\Pr[\overline{\mathcal{F}^\star}]\cdot\mathbb{E}[\max\{X_1,\dots,X_n\}\mid\overline{\mathcal{F}^\star}]+\Pr[\mathcal{F}^\star]\cdot\mathbb{E}[\max\{X_1,\dots,X_n\}\mid\mathcal{F}^\star]\\
            \geq\;&\Pr[\overline{\mathcal{F}^\star}]\cdot\mathbb{E}[\max\{X_1,\dots,X_n\}\mid\overline{\mathcal{F}^\star}]+\Pr[\mathcal{F}^\star]\cdot 2,
        \end{align*}
        we apply $X_\sigma\leq\max\{X_1,\dots,X_n\}$ and obtain
        \begin{align*}
            \mathbb{E}[X_\sigma]&\leq \mathbb{E}[\max\{X_1,\dots,X_n\}]-(\delta_2-o(1))\\
            &<\left(1-\frac{\delta_2}{3}-o(1)\right)\cdot \mathbb{E}[\max\{X_1,\dots,X_n\}].
        \end{align*}
			\item[Case 2b:] It holds that $\mathbb{E}[|U_2|]>\frac{n}{200}$. We define two events:
			\begin{itemize}
			    \item $\mathcal{G}_1$ is the event that \alg performs precisely two positive tests, namely on the values $1$ and $2$; call the corresponding realizations $X_a$ and $X_b$, respectively,
			    \item $\mathcal{G}_2$ is the event that $X_a=3$ and $X_b=2$.
			\end{itemize}
			We denote by $\mathcal{G}^\star$ the intersection of these two events. Then
			\begin{align*}
			    \Pr[\mathcal{G}^\star]&=\Pr[\mathcal{G}_1]\cdot\Pr[\mathcal{G}_2\mid \mathcal{G}_1]\\
			    &>\frac{1}{100}\cdot\left(1-\frac3n\right)^n\cdot\frac13\cdot \frac12>\delta_3
			\end{align*}
			for some constant $\delta_3>0$, where we again bound all the probabilities separately for the first inequality. For the first probability, we use that any test on the value $2$ is successful with probability $2/n$ and that each of the less than $n$ subsequent tests is negative with probability at least $1-3/n$ independently. For the second probability, we observe that $\Pr[X_i=3\mid X_i\geq 1]=1/3$ and $\Pr[X_i=2\mid X_i\geq 2]=1/2$ for any of the boxes independently. Observing that
        \begin{align*}
            \mathbb{E}[X_\sigma]&=\Pr[\overline{\mathcal{G}^\star}]\cdot\mathbb{E}[X_\sigma\mid\overline{\mathcal{G}^\star}]+\Pr[\mathcal{G}^\star]\cdot\mathbb{E}[X_\sigma\mid\mathcal{G}^\star]\\
            &=\Pr[\overline{\mathcal{G}^\star}]\cdot\mathbb{E}[X_\sigma\mid\overline{\mathcal{G}^\star}]+\Pr[\mathcal{G}^\star]\cdot 2
        \end{align*}
        while
        \begin{align*}
            &\mathbb{E}[\max\{X_1,\dots,X_n\}]\\
            =\;&\Pr[\overline{\mathcal{G}^\star}]\cdot\mathbb{E}[\max\{X_1,\dots,X_n\}\mid\overline{\mathcal{G}^\star}]+\Pr[\mathcal{G}^\star]\cdot\mathbb{E}[\max\{X_1,\dots,X_n\}\mid\mathcal{G}^\star]\\
            \geq\;&\Pr[\overline{\mathcal{G}^\star}]\cdot\mathbb{E}[\max\{X_1,\dots,X_n\}\mid\overline{\mathcal{G}^\star}]+\Pr[\mathcal{G}^\star]\cdot 3,
        \end{align*}
        we apply $X_\sigma\leq\max\{X_1,\dots,X_n\}$ and obtain
        \begin{align*}
            \mathbb{E}[X_\sigma]&\leq \mathbb{E}[\max\{X_1,\dots,X_n\}]-(\delta_3-o(1))\\
            &<\left(1-\frac{\delta_3}{3}-o(1)\right)\cdot \mathbb{E}[\max\{X_1,\dots,X_n\}].
        \end{align*}
		\end{enumerate}
    \end{enumerate}
    This completes the proof.
\end{proof}

\input{fig_comp}

We have verified numerically (by solving the dynamic program from Theorem~\ref{thm:dynprog}) that for this distribution the achievable competitive ratio decreases in $n$ in the interval $n=2,\ldots,1000$. For $n=1000$, the optimal competitive ratio is ca.\ $0.9799$ (computed with full precision). See Figure~\ref{fig:comp} for the results.

\section{Multiple Tests per Box}

In this section, we consider a setting with $n$ boxes and a budget of $n$ threshold tests. Each box can be tested an arbitrary number of times with different thresholds\footnote{Recall that nature draws \emph{initially a single} value $X_i \sim F$ inside each box $i$. All tests on the same box are evaluated accordingly. In other words, the results of multiple tests on the same box are always consistent with the single unknown $X_i$ drawn upfront.} as long as there are still tests available. We again assume continuous distributions and show the following result.

\begin{theorem}
    \label{thm:multiTest}
    There is an efficient $(1-o(1))$-competitive algorithm for threshold testing with multiple tests per box and a continuous distribution.
\end{theorem}

\begin{proof}
In the first step, our algorithm discards the last $\lceil n^{2/3}\rceil$ boxes, losing only a $\lceil n^{2/3}\rceil/n$ fraction of the value. The remaining ones are tested for the threshold $F^{-1}(1-n^{-2/3})$. Let $P$ be the set of boxes that were tested positively. For each box $i \in P$, the algorithm next searches the integers $\{0,\dots,\lfloor n^{4/3}\rfloor\}$ to find the largest $j$ from this set such that the test for $F^{-1}(1-n^{-2/3}+j\cdot n^{-2})$ is positive. Then\footnote{We use the convention $F^{-1}(x)=F^{-1}(1)$ for $x>1$.},
$F^{-1}(1-n^{-2/3}+j/n^2) \leq X_i < F^{-1}(1-n^{-2/3}+(j+1)/n^2)$. Using a binary search, this requires $O(\log n)$ tests for each box $i \in P$. Using the result, we say that box $i$ is of type $j$. Since there are potentially up to $n$ boxes in $P$, the algorithm may well run out of tests during this process. Eventually, if the algorithm succeeds to determine the type of each box in $P$, it picks a box from $P$ with the highest type. If no such box exists, it is not unique, or the algorithm ran out of tests before determining the type of each box in $P$, it may choose an arbitrary box. 

To analyze our algorithm, we fix any $v\in[F^{-1}(1-n^{-2/3}),F^{-1}(1))$. We denote by $\calM_v$ the event that $\max\{X_1,\dots,X_n\}=v$, and by $X_\sigma$ the value obtained by the algorithm. Our goal is to show that, whenever an optimal box has such a \emph{high} value $v$, the probability that we choose an optimal box is
\begin{equation}\label{eq:mult-test-goal}
    \Pr[X_\sigma=v\mid\mathcal{M}_v]=1-o(1).
\end{equation}
Now we see an optimal box with a high value with probability
\[\Pr[\max\{X_1,\dots,X_n\} \geq F^{-1}(1-n^{-2/3})] = 1 - (1-n^{-2/3})^n=1-o(1),\] 
so proving Eq.~\eqref{eq:mult-test-goal} indeed suffices to prove the theorem. To show Eq.~\eqref{eq:mult-test-goal}, we define two additional events:
\begin{itemize}
    \item $\calE_1$ is the event that $|P|\leq n^{1/2}$ (in particular, this implies that the algorithm does not run out of tests for large-enough $n$),
    \item $\calE_2$ is the event that only a single box has the largest type.
\end{itemize}
Note that $\Pr[X_\sigma=v\mid\mathcal{M}_v]\geq \Pr[\calE_1\cap \calE_2\mid \mathcal{M}_v]$ since our algorithm chooses the box with the maximum value if $\mathcal{M}_v$, $\calE_1$, and $\calE_2$ occur. We finalize the argument by observing that, for large-enough $n$,
\begin{align*}
    \Pr[\calE_1\cap \calE_2\mid \mathcal{M}_v]&=\Pr[\calE_1\mid \mathcal{M}_v]\cdot\Pr[\calE_2\mid \mathcal{M}_v\cap\calE_1]\\
    &\geq \left(1-e^{-\frac{n^{1/3}}{3}}\right)\cdot \left(1-\frac{n^{-2}}{1-n^{-2/3}}\right)^{n}= 1-o(1).
\end{align*}
To see the inequality, we first notice that, for large-enough $n$, the event $\calE_1$ occurs if the number of boxes in $P$ apart from the one with realization $v$ is at most $2 n^{1/3}$. This allows us to bound the first probability using a one-sided multiplicative Chernoff bound with expected value $n^{1/3}$ and factor $2$. We bound the second probability by observing that, conditioned on $\calM_v$, the probability that a single box other than that with realization $v$ has the same type is at most $n^{-2}/(1-n^{-2/3})$. Here, $n^{-2}$ is an upper bound on the probability that an independently drawn value has any fixed type, and $1-n^{-2/3}$ is a lower bound on the probability that such a value is below $v$ (using that $v$ is a high value). The additional condition on $\calE_1$ does not increase the probability of $\mathcal{E}_2$. 
This shows Eq.~\eqref{eq:mult-test-goal} and thus completes the proof.
\end{proof}

It is rather straightforward to apply the insights from Sec.~\ref{sec:discrete} to show similar results for testing with multiple tests per box and a finite discrete distribution. Our algorithm in the proof of Theorem~\ref{thm:multiTest} can be cast as a \emph{sequential} testing algorithm: It tests boxes sequentially from box $1$ to $n-\lceil n^{2/3}\rceil$. For each box $i$ it applies tests to determine whether $i \in P$ or not, and then binary search the type of $i$ (or aborts when it runs out of tests). For finite discrete distributions, we can optimize over such sequential testing algorithms using backwards induction, much like in the proof of Theorem~\ref{thm:dynprog}. When considering box $i$, an optimal decision about the next test can be found by relying on three additional parameters. Apart from the best conditional expectation of a previous box $V_{i-1}^*$, we also consider the smallest realization for which we saw a positive test for $i$, the largest one for which we saw a negative test for $i$, as well as the number of tests we applied so far. These parameters sufficiently describe the current state of the system before applying the next test. Note that there is only a polynomial number of combinations of these parameters that need to be considered. Then the algorithm has up to $m$ possible options for the next test of box $i$ -- or $m$ possible options for the first test of box $i+1$, thereby concluding the testing of box~$i$. Hence, there are only polynomially many combinations that need to be considered to find the optimal decision for the current test (assuming that an optimal testing algorithm for the subsequent number of tests/boxes has already been computed via backwards induction).

We can also transfer the approximation guarantee for the algorithm from Theorem~\ref{thm:multiTest}. We apply the algorithm in the model with probability tests and interpret them as randomized threshold tests. By applying the arguments of Proposition~\ref{prop:opt-real} to the sequential model with multiple tests per box, we see that for every randomized threshold testing algorithm there is a deterministic one that performs at least as good. Overall, this yields the following corollary.

\begin{corollary}
    For finite discrete distributions, an optimal sequential testing algorithm for multiple tests per box can be computed by dynamic programming in polynomial time. It is at least $(1-o(1))$-competitive for threshold testing with multiple tests per box.
\end{corollary}

\section{Conclusion}

In this paper, we have initiated the study of threshold testing of i.i.d.\ random variables, a probing model with partial revelation and binary feedback. For non-adaptive algorithms, the model is essentially equivalent to the standard gambler's problem, and optimal performance is governed by the i.i.d.\ prophet inequality of approximately $0.745$. For adaptive algorithms, we obtain a testing algorithm with competitive ratio of $0.869$. This significantly outperforms $0.745$, proves that there is a substantial adaptivity gap, and reveals the structural difference of the adaptive problem. Moreover, we show a constant upper bound on the ratio achievable by any adaptive testing algorithm. In contrast, when we can (adaptively) apply multiple tests to a single box, it is possible to achieve even a ratio of $1-o(1)$.

There are many intriguing open problems arising from our work. Obviously, the current upper and lower bounds for the i.i.d.\ model are not tight. More generally, a simple argument similar to Observation~\ref{obs:non-adaptive} shows that free-order prophet inequalities~\cite{journals/corr/abs-2211-04145} transfer directly to non-adaptive threshold testing, even for non-i.i.d.\ boxes. It is an intriguing open problem whether these guarantees can be strictly improved using an adaptive testing algorithm. Can we obtain a ratio strictly larger than $0.745$ also for non-i.i.d.\ threshold testing?

In addition, there are many combinatorial versions of the problem that deserve attention, i.e., when the algorithm is allowed to select more than one box. Testing algorithms for, e.g., knapsack, matroid, or general downward-closed feasibility structures represent a natural and important direction for future research.

\paragraph*{Acknowledgements}

We gratefully acknowledge discussions with Daniel Schmand and Luca von der Brelie. We also thank Rishabh Dhiman, Divyanshu Agarwal, Ashish Chiplunkar, Utsav Krishna Jaiswal, and an anonymous reviewer for pointing out errors in the parameter choice in an earlier version of the proof of Theorem 4.

\bibliography{testing}

\begin{thebibliography}{10}

\bibitem{AdamczykSW14}
Marek Adamczyk, Maxim Sviridenko, and Justin Ward.
\newblock Submodular stochastic probing on matroids.
\newblock In {\em Symposium on Theoretical Aspects of Computer Science
  (STACS)}, pages 29--40, 2014.

\bibitem{AgrawalSZ19}
S.~Agrawal, J.~Sethuraman, and X.~Zhang.
\newblock On optimal ordering in the optimal stopping problem.
\newblock In {\em ACM Conference on Economics and Computation (EC)}, pages
  187--188, 2020.

\bibitem{Alaei14}
Saeed Alaei.
\newblock Bayesian combinatorial auctions: Expanding single buyer mechanisms to
  many buyers.
\newblock {\em {SIAM} J. Comput.}, 43(2):930--972, 2014.

\bibitem{AsadpourNS08}
Arash Asadpour, Hamid Nazerzadeh, and Amin Saberi.
\newblock Stochastic submodular maximization.
\newblock In {\em Workshop on Internet and Network Economics (WINE)}, pages
  477--489, 2008.

\bibitem{journals/corr/abs-2211-04145}
Archit Bubna and Ashish Chiplunkar.
\newblock Prophet inequality: Order selection beats random order.
\newblock {\em CoRR}, abs/2211.04145, 2022.

\bibitem{abs-2011-06516}
Jos{\'{e}}~R. Correa, Andr{\'{e}}s Cristi, Boris Epstein, and Jos{\'{e}}~A.
  Soto.
\newblock Sample-driven optimal stopping: From the secretary problem to the
  i.i.d. prophet inequality.
\newblock {\em CoRR}, abs/2011.06516, 2020.

\bibitem{CorreaDFS22}
Jos{\'{e}}~R. Correa, Paul D{\"{u}}tting, Felix~A. Fischer, and Kevin Schewior.
\newblock Prophet inequalities for independent and identically distributed
  random variables from an unknown distribution.
\newblock {\em Math. Oper. Res.}, 47(2):1287--1309, 2022.

\bibitem{CorreaDFSZ21}
Jos{\'{e}}~R. Correa, Paul D{\"{u}}tting, Felix~A. Fischer, Kevin Schewior, and
  Bruno Ziliotto.
\newblock Unknown {I.I.D.} prophets: Better bounds, streaming algorithms, and a
  new impossibility (extended abstract).
\newblock In {\em Innovations in Theoretical Computer Science Conference
  (ITCS)}, pages 86:1--86:1, 2021.

\bibitem{CorreaFHOV18}
Jos{\'{e}}~R. Correa, Patricio Foncea, Ruben Hoeksma, Tim Oosterwijk, and Tjark
  Vredeveld.
\newblock Recent developments in prophet inequalities.
\newblock {\em SIGecom Exch.}, 17(1):61--70, 2018.

\bibitem{CorreaFHOV21}
Jos{\'{e}}~R. Correa, Patricio Foncea, Ruben Hoeksma, Tim Oosterwijk, and Tjark
  Vredeveld.
\newblock Posted price mechanisms and optimal threshold strategies for random
  arrivals.
\newblock {\em Math. Oper. Res.}, 46(4):1452--1478, 2021.

\bibitem{CorreaSZ21}
Jos{\'{e}}~R. Correa, Raimundo Saona, and Bruno Ziliotto.
\newblock Prophet secretary through blind strategies.
\newblock {\em Math. Program.}, 190(1):483--521, 2021.

\bibitem{DuttingFKL20}
Paul D{\"{u}}tting, Michal Feldman, Thomas Kesselheim, and Brendan Lucier.
\newblock Prophet inequalities made easy: Stochastic optimization by pricing
  nonstochastic inputs.
\newblock {\em {SIAM} J. Comput.}, 49(3):540--582, 2020.

\bibitem{conf/icalp/0002HKSV14}
Oliver G{\"{o}}bel, Martin Hoefer, Thomas Kesselheim, Thomas Schleiden, and
  Berthold V{\"{o}}cking.
\newblock Online independent set beyond the worst-case: Secretaries, prophets,
  and periods.
\newblock In {\em International Colloqium on Automata, Languages, and
  Programming (ICALP)}, pages 508--519, 2014.

\bibitem{GuptaJSS19}
Anupam Gupta, Haotian Jiang, Ziv Scully, and Sahil Singla.
\newblock The markovian price of information.
\newblock In {\em Integer Programming and Combinatorial Optimization (IPCO)},
  pages 233--246, 2019.

\bibitem{GuptaN13}
Anupam Gupta and Viswanath Nagarajan.
\newblock A stochastic probing problem with applications.
\newblock In {\em Integer Programming and Combinatorial Optimization (IPCO)},
  pages 205--216, 2013.

\bibitem{GuptaNS16}
Anupam Gupta, Viswanath Nagarajan, and Sahil Singla.
\newblock Algorithms and adaptivity gaps for stochastic probing.
\newblock In {\em {ACM-SIAM} Symposium on Discrete Algorithms (SODA)}, pages
  1731--1747, 2016.

\bibitem{GuptaNS17}
Anupam Gupta, Viswanath Nagarajan, and Sahil Singla.
\newblock Adaptivity gaps for stochastic probing: Submodular and {XOS}
  functions.
\newblock In {\em {ACM-SIAM} Symposium on Discrete Algorithms (SODA)}, pages
  1688--1702. {SIAM}, 2017.

\bibitem{HajiaghayiKS07}
Mohammad~Taghi Hajiaghayi, Robert~D. Kleinberg, and Tuomas Sandholm.
\newblock Automated online mechanism design and prophet inequalities.
\newblock In {\em {AAAI} Conference on Artificial Intelligence (AAAI)}, pages
  58--65, 2007.

\bibitem{HillK82}
T.~P. Hill and R.~P. Kertz.
\newblock Comparisons of stop rule and supremum expectations of i.i.d. random
  variables.
\newblock {\em Annals of Probability}, 10(2):336--345, 1982.

\bibitem{conf/ijcai/0001SS21}
Martin Hoefer, Kevin Schewior, and Daniel Schmand.
\newblock Stochastic probing with increasing precision.
\newblock In {\em International Joint Conference on Artificial Intelligence
  (IJCAI)}, pages 4069--4075, 2021.

\bibitem{KleinbergW19}
Robert Kleinberg and S.~Matthew Weinberg.
\newblock Matroid prophet inequalities and applications to multi-dimensional
  mechanism design.
\newblock {\em Games Econ. Behav.}, 113:97--115, 2019.

\bibitem{KrengelS77}
U.~Krengel and L.~Sucheston.
\newblock Semiamarts and finite values.
\newblock {\em Bull. Amer. Math. Soc.}, 83:745--747, 1977.

\bibitem{LemeSTW23}
Renato~Paes Leme, Balasubramanian Sivan, Yifeng Teng, and Pratik Worah.
\newblock Pricing query complexity of revenue maximization.
\newblock In {\em {ACM-SIAM} Symposium on Discrete Algorithms (SODA)}, pages
  399--415. SIAM, 2023.

\bibitem{journals/corr/abs-2205-05519}
Bo~Li, Xiaowei Wu, and Yutong Wu.
\newblock Query efficient prophet inequality with unknown {I.I.D.}
  distributions.
\newblock {\em CoRR}, abs/2205.05519, 2022.

\bibitem{LiuLPSS21}
Allen Liu, Renato~Paes Leme, Martin P{\'{a}}l, Jon Schneider, and
  Balasubramanian Sivan.
\newblock Variable decomposition for prophet inequalities and optimal ordering.
\newblock In {\em {ACM} Conference on Economics and Computation (EC)}, page
  692, 2021.

\bibitem{Lucier17}
Brendan Lucier.
\newblock An economic view of prophet inequalities.
\newblock {\em SIGecom Exch.}, 16(1):24--47, 2017.

\bibitem{MeisterN21}
Michela Meister and Sloan Nietert.
\newblock Learning with comparison feedback: Online estimation of sample
  statistics.
\newblock In {\em Algorithmic Learning Theory (ALT)}, volume 132, pages
  983--1001, 2021.

\bibitem{OkoroaforGK23}
Princewill Okoroafor, Vaishnavi Gupta, and Robert Kleinberg.
\newblock Non-stochastic {CDF} estimation using threshold queries.
\newblock In {\em {ACM-SIAM} Symposium on Discrete Algorithms (SODA)}, pages
  3551--3572. SIAM, 2023.

\bibitem{PengT22}
Bo~Peng and Zhihao~Gavin Tang.
\newblock Order selection prophet inequality: From threshold optimization to
  arrival time design.
\newblock In {\em {IEEE} Annual Symposium on Foundations of Computer Science
  (FOCS)}, pages 171--178, 2022.

\bibitem{journals/corr/abs-2210-05634}
Sebastian Perez{-}Salazar, Mohit Singh, and Alejandro Toriello.
\newblock The {IID} prophet inequality with limited flexibility.
\newblock {\em CoRR}, abs/2210.05634, 2022.

\bibitem{RubinsteinWW20}
Aviad Rubinstein, Jack~Z. Wang, and S.~Matthew Weinberg.
\newblock Optimal single-choice prophet inequalities from samples.
\newblock In {\em Innovations in Theoretical Computer Science Conference
  (ITCS)}, pages 60:1--60:10, 2020.

\bibitem{Weitzman79}
Martin~L. Weitzman.
\newblock Optimal search for the best alternative.
\newblock {\em Econometrica}, 47:641--654, 1979.

\end{thebibliography}

\appendix

\newpage

\section{Missing Proofs}
\label{app:proofs}

\subsection{Proof of Theorem~\ref{thm:iid-alg}}
\subparagraph{Three Thresholds.} 
We first analyze the probabilities of obtaining a box of sufficient value. As above
\begin{align*}
    \Pr[\calE_{10}] = \alpha_1 \cdot \frac{\left(1-\frac{\alpha_2}{n}\right)^n - \left(1-\frac{\alpha_1}{n}\right)^n}{\alpha_1 - \alpha_2} \enspace.
\end{align*}

We call $\calE_{110}$ the event that we have a positive test for $\tau_1$, $\tau_2$ but not for $\tau_3$, which happens with probability
\begin{align*}
\Pr[\calE_{110}] &= \sum_{\ell_1=0}^{n-2} (1-q_1)^{\ell_1} q_1 \cdot \left(\sum_{\ell_2 = \ell_1}^{n-2} (1-q_2)^{\ell_2-\ell_1} q_2 \cdot (1-q_3)^{n-2-\ell_2}\right) \\
&= q_1 q_2 \cdot (1-q_3)^{n-2} \cdot \sum_{\ell_1=0}^{n-2} \left(\frac{1-q_1}{1-q_2}\right)^{\ell_1} \cdot \sum_{\ell_2 = \ell_1}^{n-2} \left(\frac{1-q_2}{1-q_3}\right)^{\ell_2} \\
&= q_1 q_2 \cdot (1-q_3)^{n-1} \cdot \sum_{\ell_1=0}^{n-2} \left(\frac{1-q_1}{1-q_2}\right)^{\ell_1} \cdot \frac{1}{q_2 - q_3} \cdot \left( \left(\frac{1-q_2}{1-q_3}\right)^{\ell_1} - \left(\frac{1-q_2}{1-q_3}\right)^{n-1}\right) \\
&= q_1 q_2 \cdot \frac{(1-q_3)^{n-1}}{q_2 - q_3} \cdot \sum_{\ell_1=0}^{n-2} \left(\frac{1-q_1}{1-q_3}\right)^{\ell_1} \cdot \left( 1 - \left(\frac{1-q_2}{1-q_3}\right)^{n-1-\ell_1}\right) \\
&= q_1 q_2 \cdot \frac{(1-q_3)^{n-1}}{q_2 - q_3} \cdot \left(\frac{1- \left(\frac{1-q_1}{1-q_3}\right)^{n-1}}{1-\frac{1-q_1}{1-q_3}} - \left(\frac{1-q_2}{1-q_3}\right)^{n-1} \sum_{\ell_1=0}^{n-2} \left(\frac{1-q_1}{1-q_2}\right)^{\ell_1}\right)\\
&= q_1 q_2 \cdot \frac{(1-q_3)^{n-1}}{q_2 - q_3}  \left(\frac{1 - \left(\frac{1-q_1}{1-q_3}\right)^{n-1}}{1-\frac{1-q_1}{1-q_3}} - (1-q_2) \frac{(1-q_2)^{n-1}- (1-q_1)^{n-1}}{(1-q_3)^{n-1}\cdot (q_1 - q_2)}\right)\\
&= q_1 q_2 \cdot \left((1-q_3)\frac{(1-q_3)^{n-1} - (1-q_1)^{n-1}}{(q_1-q_3)(q_2-q_3)} - (1-q_2) \frac{(1-q_2)^{n-1}- (1-q_1)^{n-1}}{(q_1 - q_2)(q_2-q_3)}\right)\\
&= q_1 q_2 \cdot \frac{(q_2 - q_3)(1-q_1)^n - (q_1-q_3)(1-q_2)^n + (q_1-q_2)(1-q_3)^n}{(q_1-q_2)(q_1-q_3)(q_2-q_3)} \\
&= \alpha_1 \alpha_2 \cdot \frac{(\alpha_2 - \alpha_3)\left(1-\frac{\alpha_1}{n}\right)^n - (\alpha_1-\alpha_3)\left(1-\frac{\alpha_2}{n}\right)^n + (\alpha_1-\alpha_2)\left(1-\frac{\alpha_3}{n}\right)^n}{(\alpha_1-\alpha_2)(\alpha_1-\alpha_3)(\alpha_2-\alpha_3)} 
\end{align*}

We call $\calE_{111}$ the event that we have positive tests for $\tau_1$, $\tau_2$ and $\tau_3$. Moreover, $\calE_{11}$ is the event that we have a positive test for $\tau_1$ and $\tau_2$ (irrespective of what happens subsequently for $\tau_3$). 
\begin{align*}
\end{align*}

Now for $\alpha \in I_3$ we have three possibilities. In case event $\calE_{10}$ happens, the probability $A(\alpha)$ is $\alpha/\alpha_1$, for $\calE_{110}$ it is $\alpha/\alpha_2$ and for $\calE_{111}$ it is $\alpha/\alpha_3$. Overall
\begin{align*}
&A(\alpha) = \frac{\alpha}{\alpha_1} \Pr[\calE_{10}] + \frac{\alpha}{\alpha_2} \Pr[\calE_{110}] + \frac{\alpha}{\alpha_3} \Pr[\calE_{111}]\\
 &= \alpha \cdot \left(\frac{\Pr[E_{10}]}{\alpha_1} + \frac{\Pr[E_{110}]}{\alpha_2}  + \frac{\Pr[\calE_1] - \Pr[\calE_{10}] - \Pr[\calE_{110}]}{\alpha_3}\right)\\
 &= \frac{\alpha}{\alpha_3} \cdot \left(\Pr[E_1] - \left(\frac{\alpha_1 - \alpha_3}{\alpha_1}\right) \Pr[\calE_{10}] - \left(\frac{\alpha_2 - \alpha_3}{\alpha_2}\right) \Pr[\calE_{110}]\right)\\
 &= \frac{\alpha}{\alpha_3} \cdot \Bigg(1 - \left(1-\frac{\alpha_1}{n}\right)^n - \frac{\alpha_1 - \alpha_3}{\alpha_1-\alpha_2} \left(\left(1-\frac{\alpha_2}{n}\right)^n - \left(1-\frac{\alpha_1}{n}\right)^n\right)\\
 & \qquad- \frac{\alpha_1}{(\alpha_1-\alpha_2)(\alpha_1-\alpha_3)}\cdot \\
 & \qquad \quad \left((\alpha_2 - \alpha_3)\left(1-\frac{\alpha_1}{n}\right)^n - (\alpha_1-\alpha_3)\left(1-\frac{\alpha_2}{n}\right)^n + (\alpha_1-\alpha_2)\left(1-\frac{\alpha_3}{n}\right)^n\right)\Bigg) \\
 &= \frac{\alpha}{\alpha_3} \cdot \Bigg(1 - \left(1 - \frac{\alpha_1-\alpha_3}{\alpha_1-\alpha_2} + \frac{\alpha_1(\alpha_2-\alpha_3)}{(\alpha_1-\alpha_2)(\alpha_1-\alpha_3)}\right) \left(1-\frac{\alpha_1}{n}\right)^n \\
 & \qquad- \left(\frac{q_1-q_3}{q_1-q_2} - \frac{q_1}{q_1-q_2}\right) (1-q_2)^n - \frac{q_1}{q_1-q_3}(1-q_3)^n \Bigg)\\
 &= \frac{\alpha}{\alpha_3} \cdot \Bigg(1 - \left(1 + \frac{\alpha_3}{\alpha_1-\alpha_2} - \frac{\alpha_1}{\alpha_1-\alpha_3}\right) \left(1-\frac{\alpha_1}{n}\right)^n + \frac{\alpha_3}{\alpha_1-\alpha_2} \left(1-\frac{\alpha_2}{n}\right)^n \\
 &\qquad- \frac{\alpha_1}{\alpha_1-\alpha_3}\left(1-\frac{\alpha_3}{n}\right)^n \Bigg)\\
 &= \frac{\alpha}{\alpha_3} \cdot \Bigg(1 - \left(1-\frac{\alpha_1}{n}\right)^n + \frac{\alpha_3}{\alpha_1-\alpha_2} \left(\left(1-\frac{\alpha_2}{n}\right)^n - \left(1-\frac{\alpha_1}{n}\right)^n\right)\\
 &\qquad- \frac{\alpha_1}{\alpha_1-\alpha_3} \left(\left(1-\frac{\alpha_3}{n}\right)^n - \left(1-\frac{\alpha_1}{n}\right)^n\right)\Bigg)\\
 &= c_3(t) \cdot\left(1 - \left(1-\frac{\alpha}{n}\right)^n\right)
\end{align*}
Similarly as for $c_2(\alpha)$ for two thresholds, we observe that the dependence of $c_3(\alpha)$ on $\alpha$ is fully captured by the term $\alpha/(1-(1-\alpha/n)^n)$. For every given $n \ge 1$ and every $\alpha > 0$, this ratio is at least 1 as argued above. Hence, to obtain a lower bound we replace this term by 1. Applying a limit we see that
\[
\lim_{n \to \infty} c_3(\alpha) \ge \frac{\alpha}{\alpha_3} \cdot \frac{1}{1-\e^{-\alpha}} \cdot \left(1 - \e^{-\alpha_1} + \frac{\alpha_3 (\e^{-\alpha_2} - \e^{-\alpha_1})}{\alpha_1 - \alpha_2} - \frac{\alpha_1 (\e^{-\alpha_3} - \e^{-\alpha_1})}{\alpha_1 - \alpha_3} \right) 
\]

For $\alpha \in I_2$ we have
\begin{align*}
&A(\alpha) = \frac{\alpha}{\alpha_1} \Pr[\calE_{10}] + \frac{\alpha}{\alpha_2} \Pr[\calE_{110}] + \Pr[\calE_{111}]\\
&= \frac{\alpha}{\alpha_1} \Pr[\calE_{10}] + \frac{\alpha}{\alpha_2} \Pr[\calE_{110}] + (\Pr[\calE_1] - \Pr[\calE_{10}] - \Pr[\calE_{110}]\\
&= \Pr[\calE_1] - \frac{\alpha_1-\alpha}{\alpha_1} \Pr[\calE_{10}] - \frac{\alpha_2-\alpha}{\alpha_2} \Pr[\calE_{110}] \\
&= 1 - \left(1-\frac{\alpha_1}{n}\right)^n - \frac{\alpha_1-\alpha}{\alpha_1-\alpha_2}\left(\left(1-\frac{\alpha_2}{n}\right)^n - \left(1-\frac{\alpha_1}{n}\right)^n\right)\\
& \qquad- \frac{\alpha_1(\alpha_2-\alpha)}{(\alpha_1-\alpha_2)(\alpha_1-\alpha_3)(\alpha_2-\alpha_3)} \cdot \\
&\qquad \quad \left((\alpha_2 - \alpha_3)\left(1-\frac{\alpha_1}{n}\right)^n - (\alpha_1-\alpha_3)\left(1-\frac{\alpha_2}{n}\right)^n + (\alpha_1-\alpha_2)\left(1-\frac{\alpha_3}{n}\right)^n\right) \\
&= 1 - \left(1-\frac{\alpha_1-\alpha}{\alpha_1-\alpha_2} + \frac{\alpha_1(\alpha_2-\alpha)}{(\alpha_1-\alpha_2)(\alpha_1-\alpha_3)}\right) \left(1-\frac{\alpha_1}{n}\right)^n \\
&\qquad- \left(\frac{\alpha_1-\alpha}{\alpha_1-\alpha_2} - \frac{\alpha_1(\alpha_2-\alpha)}{(\alpha_1-\alpha_2)(\alpha_2-\alpha_3)}\right)\left(1-\frac{\alpha_2}{n}\right)^n - \frac{\alpha_1(\alpha_2-\alpha)}{(\alpha_1-\alpha_3)(\alpha_2-\alpha_3)}\left(1-\frac{\alpha_3}{n}\right)^n\\
&= 1 - \left(1-  \frac{\alpha(\alpha_1-\alpha_2+\alpha_3) - \alpha_1\alpha_3}{(\alpha_1-\alpha_2)(\alpha_2-\alpha_3)} + \frac{\alpha \alpha_1 - \alpha_1\alpha_2}{(\alpha_1-\alpha_3)(\alpha_2-\alpha_3)}\right) \left(1-\frac{\alpha_1}{n}\right)^n \\
&\qquad- \frac{\alpha(\alpha_1-\alpha_2+\alpha_3) - \alpha_1\alpha_3}{(\alpha_1-\alpha_2)(\alpha_2-\alpha_3)}\left(1-\frac{\alpha_2}{n}\right)^n + \frac{\alpha \alpha_1 - \alpha_1\alpha_2}{(\alpha_1-\alpha_3)(\alpha_2-\alpha_3)}\left(1-\frac{\alpha_3}{n}\right)^n\\
&= c_2(\alpha) \cdot \left(1 - \left(1 - \frac{\alpha}{n}\right)^n\right)
\end{align*}
and, thus, the limit $n\to \infty$ yields 
\begin{align*}
    \lim_{n \to \infty} c_2(\alpha) &= \frac{1}{1- \e^{-\alpha}} \cdot \Bigg(1 - \e^{-\alpha_1} - \frac{(\alpha (\alpha_1 - \alpha_2 + \alpha_3) - \alpha_1\alpha_3)(\e^{-\alpha_2} - \e^{-\alpha_1})}{(\alpha_1 - \alpha_2)(\alpha_1-\alpha_3)}\\
    &\qquad + \frac{(\alpha \alpha_1 - \alpha_1\alpha_2)(\e^{-\alpha_3} - \e^{-\alpha_1})}{(\alpha_1-\alpha_3)(\alpha_2-\alpha_3)}\Bigg)
\end{align*}

For $\alpha \in I_1$ we have
\begin{align*}
    A(\alpha) &= \frac{\alpha}{\alpha_1} \Pr[\calE_{10}] + \Pr[\calE_{110}] + \Pr[\calE_{111}] = \frac{\alpha}{\alpha_1} \Pr[\calE_{10}] + \Pr[\calE_{11}]
\end{align*}
similarly as above. Hence, we obtain
\[
    \lim_{n \to \infty} c_1(\alpha) = 
    \frac{1}{1-\e^{-\alpha}} \cdot \left(1 - \e^{-\alpha_1} - \frac{(\alpha_1 - \alpha) (e^{-\alpha_2} - \e^{-\alpha_1})}{\alpha_1 - \alpha_2}\right)
\]
as above. Similarly, 
\[
    c_0(t) \ge \frac{1-\left(1-\frac{\alpha_1}{n}\right)^n}{1-\left(1-\frac{\alpha}{n}\right)^n}
\]
as above.

Again, $c_0(\alpha_1) = c_1(\alpha_1) = 1$ and $c_1(\alpha_2) = c_2(\alpha_2) = 1$, which we already observed above. In contrast, consider $\alpha = \alpha_3$. There is a positive probability that the instance has a single box with value at least $\tau_3$ and another box with value at least $\tau_2$, but less than $\tau_3$. All remaining $n-2$ boxes have value less than $\tau_1$. In this case, the algorithm sees a positive test for $\tau_1$ and one for $\tau_2$, and selects the latter box. By the i.i.d.\ assumption, the algorithm gets a value of at least $\tau_3$ only with probability 1/2. Thus, $c_2(\alpha_3) = c_3(\alpha_3) < 1$.

To obtain the best ratio, we strive to select $0 < \alpha_3 < \alpha_2 < \alpha_1$  such that
\[
    \max_{\alpha_1,\alpha_2, \alpha_3} \; \{ \min_{\alpha \in I_3} c_3(\alpha), \, \min_{\alpha \in I_2} c_2(\alpha), \, \min_{\alpha \in I_1} c_1(\alpha), \, \min_{\alpha \in I_0} c_0(\alpha) \}
\]

Again, we numerically optimized parameters $\alpha_1,\alpha_2,\alpha_3$ and used standard solver software to numerically minimize $c_i(\alpha)$ for each $i=0,1,2,3$. The lower bounds for $c_3$ and $c_0$ then amount to $c_3(0) \ge 0.8693380\ldots$ and $c_0(\infty) \ge 0.8693371\ldots$. The minimum of $\lim_{n \to \infty} c_2(\alpha)$ is located roughly at $\alpha_2^* \approx 0.1162634\ldots$ with a value for $c_2(\alpha^*_2) \ge 0.8693454\ldots$. The minimum of $\lim_{n \to \infty} c_1(\alpha)$ is located roughly at $\alpha_1^* \approx 1.0351330\ldots$ with a value for $c_2(\alpha^*_2) \ge 0.8693365\ldots$. For a plot of the ratios see Figure~\ref{fig:3plot}.
\input{fig_3plot}

\end{document}